\documentclass[submission,copyright,creativecommons]{eptcs}
\providecommand{\event}{BEAT 2014} 
\usepackage{breakurl}

\usepackage{hyperref}
\usepackage{amsfonts}
\usepackage{amsmath}
\usepackage{amssymb}
\usepackage{amsthm}
\usepackage{amsbsy}
\usepackage{enumerate}
\usepackage{stackrel}
\usepackage{bm}
\usepackage[toc,page]{appendix}
\usepackage{stmaryrd}
\usepackage{wasysym}
\usepackage{color}

\newdimen\proofrulebreadth \proofrulebreadth=.05em
\newdimen\proofdotseparation \proofdotseparation=1.25ex
\newdimen\proofrulebaseline \proofrulebaseline=2ex
\newcount\proofdotnumber \proofdotnumber=3
\let\then\relax
\def\hfi{\hskip0pt plus.0001fil}
\mathchardef\squigto="3A3B
\newif\ifinsideprooftree\insideprooftreefalse
\newif\ifonleftofproofrule\onleftofproofrulefalse
\newif\ifproofdots\proofdotsfalse
\newif\ifdoubleproof\doubleprooffalse
\let\wereinproofbit\relax
\newdimen\shortenproofleft
\newdimen\shortenproofright
\newdimen\proofbelowshift
\newbox\proofabove
\newbox\proofbelow
\newbox\proofrulename
\def\shiftproofbelow{\let\next\relax\afterassignment\setshiftproofbelow\dimen0 }
\def\shiftproofbelowneg{\def\next{\multiply\dimen0 by-1 }%
\afterassignment\setshiftproofbelow\dimen0 }
\def\setshiftproofbelow{\next\proofbelowshift=\dimen0 }
\def\setproofrulebreadth{\proofrulebreadth}

\def\prooftree{
\ifnum  \lastpenalty=1
\then   \unpenalty
\else   \onleftofproofrulefalse
\fi
\ifonleftofproofrule
\else   \ifinsideprooftree
        \then   \hskip.5em plus1fil
        \fi
\fi
\bgroup
\setbox\proofbelow=\hbox{}\setbox\proofrulename=\hbox{}%
\let\justifies\proofover\let\leadsto\proofoverdots\let\Justifies\proofoverdbl
\let\using\proofusing\let\[\prooftree
\ifinsideprooftree\let\]\endprooftree\fi
\proofdotsfalse\doubleprooffalse
\let\thickness\setproofrulebreadth
\let\shiftright\shiftproofbelow \let\shift\shiftproofbelow
\let\shiftleft\shiftproofbelowneg
\let\ifwasinsideprooftree\ifinsideprooftree
\insideprooftreetrue
\setbox\proofabove=\hbox\bgroup$\displaystyle 
\let\wereinproofbit\prooftree
\shortenproofleft=0pt \shortenproofright=0pt \proofbelowshift=0pt
\onleftofproofruletrue\penalty1
}

\def\eproofbit{
\ifx    \wereinproofbit\prooftree
\then   \ifcase \lastpenalty
        \then   \shortenproofright=0pt  
        \or     \unpenalty\hfil         
        \or     \unpenalty\unskip       
        \else   \shortenproofright=0pt  
        \fi
\fi
\global\dimen0=\shortenproofleft
\global\dimen1=\shortenproofright
\global\dimen2=\proofrulebreadth
\global\dimen3=\proofbelowshift
\global\dimen4=\proofdotseparation
\global\count255=\proofdotnumber
$\egroup  
\shortenproofleft=\dimen0
\shortenproofright=\dimen1
\proofrulebreadth=\dimen2
\proofbelowshift=\dimen3
\proofdotseparation=\dimen4
\proofdotnumber=\count255
}

\def\proofover{
\eproofbit 
\setbox\proofbelow=\hbox\bgroup 
\let\wereinproofbit\proofover
$\displaystyle
}%
\def\proofoverdbl{
\eproofbit 
\doubleprooftrue
\setbox\proofbelow=\hbox\bgroup 
\let\wereinproofbit\proofoverdbl
$\displaystyle
}%
\def\proofoverdots{
\eproofbit 
\proofdotstrue
\setbox\proofbelow=\hbox\bgroup 
\let\wereinproofbit\proofoverdots
$\displaystyle
}%
\def\proofusing{
\eproofbit 
\setbox\proofrulename=\hbox\bgroup 
\let\wereinproofbit\proofusing
\kern0.3em$
}

\def\endprooftree{
\eproofbit 
  \dimen5 =0pt
\dimen0=\wd\proofabove \advance\dimen0-\shortenproofleft
\advance\dimen0-\shortenproofright
\dimen1=.5\dimen0 \advance\dimen1-.5\wd\proofbelow
\dimen4=\dimen1
\advance\dimen1\proofbelowshift \advance\dimen4-\proofbelowshift
\ifdim  \dimen1<0pt
\then   \advance\shortenproofleft\dimen1
        \advance\dimen0-\dimen1
        \dimen1=0pt
        \ifdim  \shortenproofleft<0pt
        \then   \setbox\proofabove=\hbox{%
                        \kern-\shortenproofleft\unhbox\proofabove}%
                \shortenproofleft=0pt
        \fi
\fi
\ifdim  \dimen4<0pt
\then   \advance\shortenproofright\dimen4
        \advance\dimen0-\dimen4
        \dimen4=0pt
\fi
\ifdim  \shortenproofright<\wd\proofrulename
\then   \shortenproofright=\wd\proofrulename
\fi
\dimen2=\shortenproofleft \advance\dimen2 by\dimen1
\dimen3=\shortenproofright\advance\dimen3 by\dimen4
\ifproofdots
\then
        \dimen6=\shortenproofleft \advance\dimen6 .5\dimen0
        \setbox1=\vbox to\proofdotseparation{\vss\hbox{$\cdot$}\vss}%
        \setbox0=\hbox{%
                \advance\dimen6-.5\wd1
                \kern\dimen6
                $\vcenter to\proofdotnumber\proofdotseparation
                        {\leaders\box1\vfill}$%
                \unhbox\proofrulename}%
\else   \dimen6=\fontdimen22\the\textfont2 
        \dimen7=\dimen6
        \advance\dimen6by.5\proofrulebreadth
        \advance\dimen7by-.5\proofrulebreadth
        \setbox0=\hbox{%
                \kern\shortenproofleft
                \ifdoubleproof
                \then   \hbox to\dimen0{%
                        $\mathsurround0pt\mathord=\mkern-6mu%
                        \cleaders\hbox{$\mkern-2mu=\mkern-2mu$}\hfill
                        \mkern-6mu\mathord=$}%
                \else   \vrule height\dimen6 depth-\dimen7 width\dimen0
                \fi
                \unhbox\proofrulename}%
        \ht0=\dimen6 \dp0=-\dimen7
\fi
\let\doll\relax
\ifwasinsideprooftree
\then   \let\VBOX\vbox
\else   \ifmmode\else$\let\doll=$\fi
        \let\VBOX\vcenter
\fi
\VBOX   {\baselineskip\proofrulebaseline \lineskip.2ex
        \expandafter\lineskiplimit\ifproofdots0ex\else-0.6ex\fi
        \hbox   spread\dimen5   {\hfi\unhbox\proofabove\hfi}%
        \hbox{\box0}%
        \hbox   {\kern\dimen2 \box\proofbelow}}\doll%
\global\dimen2=\dimen2
\global\dimen3=\dimen3
\egroup 
\ifonleftofproofrule
\then   \shortenproofleft=\dimen2
\fi
\shortenproofright=\dimen3
\onleftofproofrulefalse
\ifinsideprooftree
\then   \hskip.5em plus 1fil \penalty2
\fi
}

\def\titlerunning{Compliance for reversible client/server interactions
}
\def\authorrunning{Barbanera, Dezani-Ciancaglini, de'Liguoro
}

\title{Compliance for reversible client/server interactions\footnote{
 This work was partially supported by EU Collaborative project  ASCENS 257414, ICT COST Action IC1201 BETTY, MIUR PRIN Project CINA Prot. 2010LHT4KM and Torino University/Compagnia San Paolo Project SALT. 
}}

\author{Franco Barbanera
\institute{Dipartimento di Matematica e Informatica\\
University of Catania}
\email{barba@dmi.unict.it}
 \and 
Mariangiola Dezani-Ciancaglini
\institute{Dipartimento di Informatica\\
University of Torino}
\email{dezani@di.unito.it}
 \and 
Ugo de'Liguoro
\institute{Dipartimento di Informatica\\
University of Torino}
\email{deliguoro@di.unito.it}
}

\newtheorem{definition}{Definition}[section]
\newtheorem{lemma}[definition]{Lemma}
\newtheorem{proposition}[definition]{Proposition}
\newtheorem{theorem}[definition]{Theorem}

\DeclareMathAlphabet{\mathpzc}{OT1}{pzc}{m}{it}

\newcommand{\Comment}[1]{ }

\newcommand{\ckpt}{_{\tiny\mbox{$\blacktriangle$}}\!}
\newcommand{\whiteckpt}{_{\tiny\mbox{$\triangle$}}\!}
\newcommand{\Rel}{\mathpzc R}

\newcommand{\FunH}{{\cal H}}

\newcommand{\ored}[1]{\stackrel{#1}{\longrightarrow}}      
 
\newcommand{\rlbk}{\sf rbk}

\newcommand{\red}{\longrightarrow}
\newcommand{\starred}{\stackrel{*}{\red}}

\newcommand{\comply}{\dashv}
\newcommand{\complyR}{\comply^{\mbox{\tiny $\blacktriangle$}}}
\newcommand{\complyF}{\comply^{\mbox{\tiny $\prec$}}}

\newcommand{\Dual}[1]{\overline{#1}}

\newcommand{\der}{\vartriangleright}

\newcommand{\CkptcomplHyp}{\mbox{\sc Hyp}}
\newcommand{\CkptcomplAx}{\mbox{\sc Ax}}

\newcommand{\Set}[1]{\{#1\}}

\renewcommand{\implies}{\Rightarrow}

\newcommand{\lts}[1]{\stackrel{#1}{\longrightarrow}}

\newcommand{\rec}{{\sf rec} \, }

\newcommand{\sumacts}[1]{{\mathcal A^+}(#1)}
\newcommand{\oplusacts}[1]{{\mathcal A^\oplus}(#1)}

\newcommand{\Names}{{\cal N}}
\newcommand{\CoNames}{\overline{\Names}}

\newcommand{\stopA}{{\bf 1}}

\newcommand{\Sbehav}{{\sf SB}}

\newcommand{\back}{\prec}

\newcommand{\sbcb}{\Sbehav_{\tiny\back}}
\newcommand{\np}[2]{#1\back#2}
\newcommand{\p}[2]{{\mathbf b}(#1,#2)}
\newcommand{\pp}{~\|~}

\newcommand{\sbc}{\Sbehav^{\blacktriangle}}
\newcommand{\s}{\Sbehav}

\newcommand{\getsea}{\mbox{\sf sea}}
\newcommand{\getmount}{\mbox{\sf mount}}
\newcommand{\gethouse}{\mbox{\sf house}}
\newcommand{\getbung}{\mbox{\sf bung}}
\newcommand{\getgarden}{\mbox{\sf garden}}

\newcommand{\sendsea}{\overline{\getsea}}
\newcommand{\sendmount}{\overline{\getmount}}
\newcommand{\sendhouse}{\overline{\gethouse}}

\newcommand{\sendgarden}{\overline{\getgarden}}

\iftrue
\renewcommand\refname{\small REFERENCES}
\usepackage{titlesec}
\titleformat{\section}[runin]
  {\bfseries}
  {\thesection}{1em}{}[.]
\titlespacing{\section}
  {0pt}
  {0.5em}
  {0.5em}
\fi

\begin{document}

\maketitle

\begin{abstract}
In the setting of {\em session behaviours}, we study an extension of the concept of compliance when  
a disciplined form of backtracking is present. After adding checkpoints to the syntax of session behaviours, we formalise the operational semantics via a LTS, and define a natural
notion of {\em checkpoint compliance}. We then obtain a co-inductive characterisation of such compliance relation, and an axiomatic 
presentation that is proved to be sound and complete. As a byproduct we get a decision procedure for the new compliance, being the axiomatic system algorithmic.
\end{abstract}





\section{Introduction}\label{sec:intro}

In human as well as  automatic negotiations, an interesting feature is the ability of rolling back to some previous point, undoing previous choices and  possibly trying a different path. {\em Rollbacks} are familiar to the users of web browsers, and so are also the troubles that these might cause during ``undisciplined'' interactions. Clicking the ``back'' button, or going to some previous point in the chronology when we are in the middle of a transaction, say the booking of a flight, can be as smart as dangerous. In any case it is surely a behaviour that service programmers want to discipline. Also the converse has to be treated with care: a server discovering that a service becomes available after having started a conversation could take advantage from some kind of rolling backs. However, such a server would be quite unfair if  the rollbacks were completely hidden from the client.

Adding rollbacks to interaction protocols requires a sophisticated concept of client/server compliance.
In this paper we investigate protocols admitting a simple, though non trivial form of reversibility in the framework 
of the theory of contracts introduced in~\cite{CCLP06} and developed in a series of papers, e.g.~\cite{CGP10}. 
We focus here on the scenario of client/server architectures, where services stored in a repository are queried by clients to establish two-sided communications, and the central concept is that of {\em compliance}. 

More precisely, we consider the formalism of {\em session behaviours} as introduced in~\cite{BdL13,BdL10,BH13}, 
but without delegation. This is a formalism interpreting the session  types, introduced by Honda et al. in~\cite{honda.vasconcelos.kubo:language-primitives},  into a subset of CCS without $\tau$.
We extend the session behaviours syntax by means of markers that we call {\em checkpoints}; these are intended
as pointers to the last place where either the client or the server can roll back at any time. We investigate which 
constraints must be imposed to obtain a safe notion of client/server interaction in the new scenario, by defining a model
in the form of a LTS, and by characterising the resulting concept of compliance both coinductively and axiomatically.
Since the axiomatic system is algorithmic that is decidable, the compliance of behaviours with checkpoints is decidable.

Before entering into the formal development of session behaviours with checkpoints, we illustrate the basic concepts by discussing a few examples. 
Suppose that the client is a customer willing to arrange for an holiday, while the server is the web service of a travel agency.
Let the action $\getsea$ represent the quest for a seaside accommodation and let $\getmount$ stands for the request of a settlement in the mountains. By $\gethouse$ we mean the request of a house, while $\getbung$ stands for the request of a bungalow. Dual actions represent offers, so that e.g. the co-action $\Dual{\getsea}$  signals availability of accommodations in a seaside and $\Dual{\gethouse}$ that a house can be booked. 

Suppose that the customer seeks a house or a bungalow at sea, but just a house in the mountains. Then the client behaviour, represented as a process algebraic term, is described by:
\vspace{-2mm}
\[ \rho = \getsea.(\gethouse+\getbung) \ + \ \getmount.\gethouse \vspace{-2mm}
\]
where dots are sequential compositions and  sums are external choices.
We say that a client $\rho$ is {\em compliant} with a server $\sigma$, written $\rho\comply\sigma$, if all client communication actions are matched by the dual actions on the server side. 
According to this the customer will be not compliant with a server behaving as: 
\vspace{-2mm}
\[\sigma = \Dual{\getmount}.(\Dual{\gethouse} \oplus \Dual{\getbung})\vspace{-2mm}\]
where $\oplus$ is internal choice. In fact the interaction represented by the parallel composition $\rho\, \| \, \sigma$, that evolves by synchronising corresponding actions and co-actions, might lead to $\gethouse \, \| \, \Dual{\getbung}$. This means that the customer is offered a bungalow in the mountains she is not willing to reserve.

Now consider the {\em dual} behaviour of $\rho$, dubbed $\Dual{\rho}$, which is obtained by exchanging actions by the respective co-actions, and external by internal choices. Then we get the server: 
\vspace{-2mm}
\[\Dual{\rho} = \Dual{\getsea}.(\Dual{\gethouse} \oplus \Dual{\getbung}) \ \oplus \ \Dual{\getmount}.\Dual{\gethouse}
\vspace{-2mm}\] 
and clearly we get $\rho \comply \Dual{\rho}$. In general we expect that $\rho\comply\Dual{\rho}$, or equivalently that $\Dual{\sigma}\comply\sigma$, since
duality is involutive.

Taking a further step, let us consider a server such that, after sending the offer $\Dual{\getsea}$ followed by $\Dual{\gethouse}$, might realise that a better offer is now available which can be issued by sending $\Dual{\getbung}$ instead of $\Dual{\gethouse}$; this can be achieved only by rolling back to the choice $\Dual{\gethouse} \oplus \Dual{\getbung}$. Rollback is however a new feature, that cannot be easily represented by usual process algebra operations~\cite{PU07}. 

To express rollback we then introduce the symbol `$\ckpt$' to mark the point where a session behaviour can backtrack to; we call such a marker a {\em checkpoint}. We suppose that a suitable mechanism keeps memory of the past, by recording the behaviour $\ckpt\sigma$ each time the checkpoint is traversed by synchronising on some action that $\sigma$ is ready to do. For simplicity we assume that only one ``past''  can be recorded at any time, so that a new memorisation destroys the old one, leading to a model in which  the client and  the server can backtrack just to the lastly traversed checkpoint. 

By adding some checkpoints to $\Dual{\rho}$ we get for example $\sigma' = \ \ckpt (\Dual{\getsea}.\ckpt (\Dual{\gethouse} \oplus \Dual{\getbung}) \ \oplus \ \Dual{\getmount}.\Dual{\gethouse})$. With respect to $\Dual{\rho}$ the new server can undo all of the internal choices, in order to keep the negotiation open as much as possible and to give to the client some better chance for booking a place, even in case it wasn't available at the beginning of the interaction. But how should the client be redesigned to interact properly? Unfortunately the most natural choice of taking the client as the dual\linebreak {$\Dual{\sigma'} = {\ckpt (\getsea.\ckpt (\gethouse +\getbung)} \,+\, \getmount.\gethouse)$} fails. 
In fact, writing $\lts{\sf fw}$ for the forward  step and $\lts{\sf rollbk}$ for the synchronous rollback, we have among the possible interactions between $\sigma'$ and $\Dual{\sigma'}$:
\[\begin{array}{lll}
&\multicolumn{2}{l}{
\ckpt (\getsea.\ckpt (\gethouse+\getbung) \ + \ \getmount.\gethouse) \ \| \ \ckpt (\Dual{\getsea}.\ckpt (\Dual{\gethouse} \oplus \Dual{\getbung}) \ \oplus \ \Dual{\getmount}.\Dual{\gethouse})} \\
\lts{\sf fw} & \ckpt (\getsea.\ckpt (\gethouse+\getbung) \ + \ \getmount.\gethouse) \ \| \ \Dual{\getsea}.\ckpt (\Dual{\gethouse} \oplus \Dual{\getbung}) & \mbox{\em internal choice} \\ 
\lts{\sf fw} &  \ckpt (\gethouse+\getbung) \ \| \ \ckpt (\Dual{\gethouse} \oplus \Dual{\getbung}) & \mbox{\em synchronising on $\getsea$ and $\Dual{\getsea}$} \\ 
\lts{\sf fw} &  \ckpt (\gethouse+\getbung) \ \| \ \Dual{\gethouse}  & \mbox{\em internal choice} \\ 
\lts{\sf rollbk} & \ckpt (\getsea.\ckpt (\gethouse+\getbung) \ + \ \getmount.\gethouse) \ \| \ \ckpt (\Dual{\gethouse} \oplus \Dual{\getbung}) & \mbox{\em rollback to the last traversed $\ckpt$}
\end{array}\]
which is now in a stuck state. The mismatch between external and internal choice is the effect of the asymmetry of the respective semantics in process algebra. The selection of a branch in an external choice is just one step; on the contrary the synchronisation on $\Dual{\getsea}$ in the second step above comes {\em after} the internal choice has occurred. This has consequences with respect to the backtracking, since the checkpoint alignment fails.

In~\cite{BdL10} it has been proved that the dual of a server is the minimum client that complies with the server with respect to a natural (and efficiently decidable) ordering, and vice versa the dual of a client is the minimum compliant server. This is an essential feature of the theory, since it is supposed to model a scenario in which clients look for servers through a network querying a service of a certain shape, that is easier to find if we know its minimal form. To express this precisely, let us write $\rho\complyR\sigma$ to denote the compliance of $\rho$ with $\sigma$ in a setting with backtracking, that we call {\em checkpoint compliance}; then we put the requirement that in the new theory the following holds:
\vspace{-2mm}
\begin{equation}\label{eq:dualRequirement}
\forall \rho.~~~~\rho\complyR\Dual{\rho}
\vspace{-2mm}
\end{equation}

For (\ref{eq:dualRequirement}) to hold we change the operational semantics of $\oplus$ by gluing the choice and the synchronisation over a co-action, that can be formalised by the rule:
\vspace{-2mm}
\[
	\Dual{a}.\sigma_1\oplus \sigma_2  \lts{\Dual{a}} \sigma_1
\vspace{-2mm}
\]
This has however the unpleasant consequence that $a \complyR \Dual{a} \oplus \Dual{b}$, while we have that  $a \not\comply \Dual{a} \oplus \Dual{b}$, where the compliance $\comply$ is defined according to the standard LTS~\cite{BdL10,BdL13,BH13}. In general, we expect the compliance of behaviours with rollback to be conservative with respect to the compliance without rollback:
\vspace{-2mm}
\begin{equation}\label{eq:consRequirement}
\forall \rho,\sigma.~~~~\rho\complyR \sigma ~~\implies~~ \textit{erase}(\rho) \comply \textit{erase}(\sigma)
\vspace{-2mm}
\end{equation}
where $\textit{erase}$ deletes all checkpoints. We will accomplish this by asking that any co-action has a corresponding action in reducing the parallel of internal and external choices.

The essence of this change is that rolling back has to be a synchronous action, and therefore it cannot be the effect of an internal choice, since the latter is unobservable. This is a general principle. Consider the interaction
\vspace{-2mm}
\[(\getsea.\gethouse.\getmount.\gethouse) \| (\sendsea.\ckpt \sendhouse. \sendmount.\sendhouse)\vspace{-2mm}\]
It is 
the pair of a client willing to book a house at sea {\em and} a house in the mountains, and a server that 
can succeed by renting twice a house at seaside! The point is that the client has no way to be aware of what happened and to react according to her own policy, which is instead the case if both are forced to backtrack at the same time. 
For this to be guaranteed we require that the client and the server either both can or both cannot rollback in all configurations. 


We finally observe that it is not necessarily the case that compliant behaviours show some correspondence between the respective checkpoints. For example it holds that:
\vspace{-2mm}
\[ \ckpt (\getsea.\gethouse.\getgarden + \gethouse. \getgarden) \ \complyR \ \ckpt (\sendsea.\ckpt \sendhouse.\sendgarden \oplus \sendhouse.\sendgarden)
\vspace{-2mm}\]
which makes sense, since the client $\ckpt (\getsea.\gethouse.\getgarden + \gethouse. \getgarden)$ is asking for a house with garden, either at sea or anywhere else.


\section{Calculus}

As explained in the Introduction, we allow checkpoints only before internal or external choices. Therefore we define session behaviours as in~\cite{BdL13,BH13} just adding checkpointed choices. 

\begin{definition}[Session Behaviours with Checkpoints]
\label{Adef:ckpt-behav}
Let $\Names$ be some countable set of symbols and $\CoNames = \Set{\Dual{a} \mid a \in \Names}$, with
$\Names\cap\CoNames = \emptyset$. 
The set $\s$ of {\bf session behaviours with checkpoints} is defined by the 
grammar of Figure~\ref{ssbc}, 
where $I$ is non-empty and finite, the names 
and  the conames 
in choices are pairwise distinct and $\sigma$ is not a variable in $\rec x.\sigma$.
\end{definition}
\begin{figure}
\[\begin{array}{lcl@{\hspace{4mm}}l}
\sigma,\rho&:=& \mid ~ \stopA &\mbox{success} \\[1mm]
       &     & \mid ~\sum_{i\in I} a_i.\sigma_i  & \mbox{external choice} \\[1mm]
       &     & \mid ~\ckpt\sum_{i\in I} a_i.\sigma_i& \mbox{checkpointed external choice}\\[1mm]
       &     & \mid ~\bigoplus_{i\in I} \Dual{a}_i.\sigma_i & \mbox{internal choice}\\[1mm]
       &     & \mid ~\ckpt\bigoplus_{i\in I} \Dual{a}_i.\sigma_i& \mbox{checkpointed internal choice}\\[1mm]
       &     & \mid  ~x  & \mbox{variable}\\[1mm]
       &     & \mid ~\rec x. \sigma &  \mbox{recursion}
\end{array}
\]
\caption{Syntax of session behaviours with checkpoints}\label{ssbc}
\end{figure}
\noindent
Note that recursion in $\s$ is guarded and hence contractive in the usual sense. We take an equi-recursive view of recursion by equating $\rec x.\sigma  $ with $\sigma[\rec x.\sigma/x]$. Hence there is no point in considering also terms of the shape $\ckpt\rec x.\sigma$.


Let us call just {\em behaviours} the expressions in $\s$. In the operational semantics of the calculus we have to record the last encountered behaviour 
$\gamma$ that was prefixed by a checkpoint in the interaction leading to $\sigma$. Therefore we will consider configurations of the shape:
\[\np\gamma\sigma\] 
In the starting configuration or just after a rollback has occurred,  there is no further point to which the behaviour might rollback, a situation we represent by writing $\np\circ\sigma$. Let $\sbc$ be the set of behaviours starting with $\ckpt\;$; then we ask $\gamma\in \sbc\cup\{\circ\}$, which is the set of the ``pasts'', and denote by $\gamma,\delta$, possibly with superscripts, its elements.
Then the LTS of clients and servers is formalised as follows.

\begin{definition}[Reduction of Session Behaviours]\label{scs}\hfill \\
\[\begin{array}{ccc}
\np\gamma{\sum_{i\in I} a_i.\sigma_i} \ored{a_i} \np\gamma{\sigma_i}~(i\in I)~(+) &\qquad&
\np\gamma{\bigoplus_{i\in I} \Dual{a}_i.\sigma_i} \ored{\Dual{a}_i} \np\gamma{\sigma_i} ~(i\in I)~(\oplus) \\[4mm]
\prooftree \np\gamma\sigma\ored{\alpha}\np\gamma{\sigma'}\quad\alpha\in\Names\cup\CoNames
\justifies \np\gamma{{\ckpt\sigma}}\ored{\alpha}\np{{\ckpt\sigma}}{\sigma'}
\using(\blacktriangle)
\endprooftree&&
\np\sigma{\sigma'}  \lts{\rlbk}  \np\circ\sigma~(\rlbk) 
\end{array}\]
\end{definition}

\smallskip

\noindent
Notice that Rule $(+)$ is the standard forward computation for external choice, but for the presence of the $\np\gamma{\cdot}$. Rule $(\oplus)$ glues into just one step both the internal choice and the communication of a coname, becoming very similar to the rule for external choice. The reduction of client/server parallel compositions (Definition~\ref{s} below) will be only possible when all internal choices can be matched by the corresponding external choices, which has the effect of saving the conservativity principle (\ref{eq:consRequirement}) of the Introduction. 
Rule $(\blacktriangle)$ says that in the presence of a checkpoint the forward reduction must update the  behaviour at which it is possible to rollback (in this case $\ckpt\sigma$). Rule $(\rlbk)$ implements the rollback: the previous past behaviour is erased, the behaviour prefixed by the last traversed checkpoint  becomes the new past behaviour and no further rollback is allowed in the new configuration.

\smallskip

When composing in parallel clients and servers we have to consider the different nature of the reductions for internal and external choices. To this aim it is handy to collect the sets of names and conames prefixing the choices, as done in the following definition. Notice that the resulting sets only contain names, since each coname is mapped to the corresponding name.
\begin{definition}[$\sumacts{\cdot}$, $\oplusacts{\cdot}$]
\label{Adef:sumoplusacts}
Let
\[\begin{array}{rclcrclcrcl}
\sumacts{\stopA}=\sumacts{\bigoplus_{i\in I} \Dual{a}_i.\sigma_i}&=&\emptyset &~~&
\sumacts{\sum_{i\in I} a_i.\sigma_i}&=&\{a_i\mid i\in I\}
&~~&\sumacts{\ckpt\sigma}&=&\sumacts{\sigma}\\
\multicolumn{1}{l}{\text{and}}\\
\oplusacts{\stopA}=\oplusacts{\sum_{i\in I} a_i.\sigma_i}&=&\emptyset &&
\oplusacts{\bigoplus_{i\in I} \Dual{a}_i.\sigma_i}&=&\{a_i\mid i\in I\}&&
\oplusacts{\ckpt\sigma}& = &\oplusacts{\sigma}
\end{array}\]
\end{definition}
The interaction of a client with a server is modelled by the reduction of their parallel composition, that can be either
forward, consisting of CCS style synchronisations, or backward, where both behaviours synchronously go back to the respective last traversed checkpointed behaviours. 
\begin{definition}[Communication Reduction of Client and Server Pairs]\label{s}\hfill \\
\[\begin{array}{c}
\prooftree
\np\delta\rho \ored{a} \np{\delta'}{\rho'} \qquad \np\gamma\sigma \ored{\Dual{a}} \np{\gamma'}{\sigma'} 
\qquad 
\oplusacts{\sigma}\subseteq \sumacts{\rho}
\justifies
\np\delta\rho\pp \np\gamma\sigma \ored{\tau} \np{\delta'}{\rho'}\pp \np{\gamma'}{\sigma'}
\endprooftree \\[8mm]
\prooftree
\np\delta\rho \ored{\Dual{a}} \np{\delta'}{\rho'} \qquad \np\gamma\sigma \ored{a} \np{\gamma'}{\sigma'} 
\qquad 
\oplusacts{\rho}\subseteq \sumacts{\sigma}
\justifies
\np\delta\rho\pp \np\gamma\sigma \ored{\tau} \np{\delta'}{\rho'}\pp \np{\gamma'}{\sigma'}
\endprooftree\\[8mm]
\prooftree
\np\rho{\rho'} \lts{\rlbk} \np\circ\rho \qquad \np\sigma{\sigma'} \lts{\rlbk} \np\circ\sigma 
\justifies
\np\rho{\rho'}\pp \np\sigma{\sigma'} \ored{\rlbk} \np\circ\rho\pp \np\circ\sigma
\endprooftree\\[1mm]
\end{array}\]
\end{definition}
\noindent
We denote by $\starred$ the reflexive and transitive closure of forward reductions.

It is easy to verify that if $\np\circ\rho\pp \np\circ\sigma \starred\np\circ\rho'\pp \np\circ\sigma'$, then $\rho\pp\sigma$ reduces to $\rho'\pp\sigma'$ in the calculi of~\cite{BdL13,BH13}, by splitting in two steps each application of rule $(\oplus)$. If  $\rho\pp\sigma$ reduces to $\rho'\pp\sigma'$ in the calculi of~\cite{BdL13,BH13} we can find $\rho''$, $\sigma''$ such that both $\rho'\pp\sigma'$ reduces to $\rho''\pp\sigma''$ and\\ \centerline{$\np\circ\rho\pp \np\circ\sigma \starred\np\circ\rho''\pp \np\circ\sigma''$.}
We take $\rho''\pp\sigma''$ as different than $\rho'\pp\sigma'$ only in case the last applied rule is an internal choice, which in rule $(\oplus)$ is fused with the communication of the coname. 

\smallskip

The last definition makes it clear that the characterisation of compliance in the present calculus requires some care, since the last checkpointed behaviours of clients and servers must be compliant. We formalise this intuition in the next section.


\section{Compliance}

The compliance relation of session behaviour calculi  requires that whenever there is no possible reduction, then all client requests and offers are satisfied, i.e. it is $\stopA$. In presence of backward computations we have also to 
require that the client and the server either both can or both cannot reverse to their last encountered checkpoints. This leads to the  following definition, in which the set of configurations is denoted by $\sbcb$, i.e.
$\sbcb=\{\np\gamma\sigma\mid
\gamma\in \sbc\cup\{\circ\},\sigma\in\s\}
$

\begin{definition}[Checkpoint Compliance Relation $\complyR$]\label{ccr}
\begin{enumerate}[i)]
\item \label{ccr1}
Let
$\FunH: {\cal P}(\sbcb\times\sbcb)\rightarrow  {\cal P}(\sbcb\times\sbcb)$ be such that, 
 for any  $\Rel\subseteq \sbcb\times\sbcb$, we get
$(\np\delta\rho,\np\gamma\sigma) \in {\cal H}(\Rel)$ if:
\begin{enumerate}[1)]
\item\label{c1} $\np\delta\rho\pp \np\gamma\sigma\not\!\!\ored{\tau}$ implies  $\rho=\stopA$ and either $ \delta=\gamma=\circ$ or $\delta,\gamma\in\sbc$;
\item\label{c2} $\np\delta\rho\pp \np\gamma\sigma \ored{\beta} \np{\delta'}{\rho'}\pp \np{\gamma'}{\sigma'}$ implies $\np{\delta'}{\rho'}\;\Rel \;\np{\gamma'}{\sigma'}$,~~where $\beta\in \{\tau,\rlbk\}$.
\end{enumerate}
\item \label{ccr2} 
A relation $\Rel\subseteq \sbcb\times\sbcb$ is a
{\em checkpoint compliance relation\/} if 
$\Rel\subseteq \FunH(\Rel)$. 
The relation $\complyR$ is the greatest solution of the equation $X=\FunH(X)$:\\ \centerline{$\complyR~=~\nu\FunH$}
\item \label{ccr3}
We say that $\rho$ is {\em checkpoint compliant} with $\sigma$ (notation $\rho\complyR \sigma$) if ~
$\np\circ\rho\complyR \np\circ\sigma$.
\end{enumerate}
\end{definition}
Roughly, when $\rho\complyR \sigma$ holds, $\rho$ and $\sigma$ are compliant in the standard sense and they keep on being so after any possible synchronous rollback that can occur during a standard interaction. Moreover it can never be the case that one of them can perform
a rollback and the other one cannot, also when $\rho$ is in the success configuration.

\smallskip

It is easy to verify that Definition \ref{ccr}(\ref{ccr3}) satisfies the requirements (\ref{eq:dualRequirement}) and (\ref{eq:consRequirement}) discussed in the Introduction.
Namely that each session behaviour is checkpoint compliant with its dual, and that
if a client and a server are checkpoint compliant, then the client and the server obtained by erasing the checkpoints are compliant. More formally, if the $\textit{erase}(\cdot)$ mapping deletes all checkpoints:
\begin{proposition}
\begin{enumerate}
\item ~~~~$\forall \rho.~~~~\rho\complyR\Dual{\rho}$.
\item ~~~~$\forall \rho,\sigma.~~~~\rho\complyR \sigma ~~\implies~~ \textit{erase}(\rho) \comply \textit{erase}(\sigma)$.
\end{enumerate}
\end{proposition}

\noindent
In the following we will use the notation $\whiteckpt\sigma$ to represent ambiguously $\sigma$ and $\ckpt\sigma$. \\
In order to give a formal system characterising checkpoint compliance it is handy to define a function \linebreak ${\mathbf b}:\sbc\cup\{\circ\}\times\s\rightarrow\sbc\cup\{\circ\}$ which returns the second argument when it is checkpointed, and the first argument otherwise.  
Formally:
\vspace{-3mm}
\[\p\gamma{{\whiteckpt\sigma}}=\begin{cases}
 \ckpt\sigma     & \text{if }{\footnotesize\mbox{$\triangle$}}= \ \blacktriangle \\
 \gamma     & \text{otherwise.}
\end{cases}
\vspace{-1mm}\]

\noindent
Forward reduction in Definition~\ref{s} can be shortly written in terms of the function ${\mathbf b}$:

\begin{lemma}\label{simple}\hfill\\
\centerline{$\begin{array}{c}
\np\gamma{{\whiteckpt(\sum_{i\in I} a_i.{\sigma}_i)}}\ored{a_k}\np{\p\gamma{{\whiteckpt(\sum_{i\in I} a_i.{\sigma}_i)}}}{\sigma_i}.\\
\np\gamma{{\whiteckpt(\bigoplus_{i\in I} \Dual{a}_i.{\sigma}_i)}}\ored{\Dual a_k}\np{\p\gamma{{\whiteckpt(\bigoplus_{i\in I} \Dual{a}_i.{\sigma}_i)}}}{\sigma_i}.
\end{array}$}

\end{lemma}
\bigskip
We now axiomatically characterise the checkpoint compliance relation by means of a formal system, whose judgments are of the form $\Gamma\der \np\delta\rho \complyF \np\gamma\sigma$, where $\Gamma$ is an environment, i.e. a finite set $\Gamma=\{\np{\delta_i}{\rho_i} \complyF \np{\gamma_i}{\sigma_i} \}_{i\in I}$.
The rules of the formal system are given in Figure~\ref{fig:complianceSystem}, where in writing $\np\gamma\delta$ we assume that  $\delta\in\s$. We denote by $\complyF$ the formal counterpart of $\complyR$.  
We are now in place to establish the soundness and completeness of the formal system in Figure~\ref{fig:complianceSystem}.

\smallskip

\begin{theorem}[Soundness]\label{sound}\hfill\\
If ~$\Gamma\der \np\delta\rho\complyF\np\gamma\sigma$ and $\np{\delta'}{\rho'}\complyR \np{\gamma'}{\sigma'}$ for all $\np{\delta'}{\rho'}\complyF \np{\gamma'}{\sigma'}\in\Gamma$, then ~$\np\delta\rho\complyR\np\gamma\sigma$.
\end{theorem}
\begin{proof}{\em (Sketch)} By induction on derivations. If the last applied rule is $(\CkptcomplHyp)$ it is trivial. 

If the last applied rule is $(\CkptcomplAx)$ and $\delta=\gamma=\circ$, then condition~(\ref{c1}) of Definition~\ref{ccr} is satisfied and condition~(\ref{c2}) of Definition~\ref{ccr} is trivially satisfied, since there are no reductions.

If the last applied rule is $(\CkptcomplAx)$ and $\Gamma\der \np\circ\delta\complyF\np\circ\gamma$ we get $\delta,\gamma\in\s$, which implies $\delta,\gamma\in\sbc$
by construction, so condition~(\ref{c1}) of Definition~\ref{ccr} is satisfied. In this case the only possible reduction is $\np\delta\stopA \pp \np\gamma\sigma\ored{\rlbk}
\np\circ\delta\pp\np\circ\gamma$. The premise $\Gamma\der \np\circ\delta\complyF\np\circ\gamma$ implies by induction $\np\circ\delta\complyR\np\circ\gamma$, so also condition~(\ref{c2}) of Definition~\ref{ccr} is satisfied. 

If the last applied rule is $( + \cdot \oplus)$, then condition~(\ref{c1}) of Definition~\ref{ccr} is trivially satisfied. In this case by Lemma~\ref{simple} \mbox{$\np\delta\rho\pp\np\gamma\sigma\ored{\tau} \np{\p\delta\rho}{\rho_j}\pp\np{\p\gamma\sigma}{\sigma_j}$} for all $j\in J$. 
The premise\\ \centerline{$\Gamma'\der \np{\p\delta\rho}{\rho_j}\complyF \np{\p\gamma\sigma}{\sigma_j}$} gives $\np{\p\delta\rho}{\rho_j}\complyR\np{\p\gamma\sigma}{\sigma_j}$ by induction and since $\complyR$ is the greatest fix point. If $\delta=\gamma=\circ$ there is no rollback, otherwise\\ \centerline{$\np\delta\rho\pp\np\gamma\sigma\ored{\rlbk}
\np\circ\delta\pp\np\circ\gamma$.} The premise $\Gamma'\der \np\circ\delta\complyF\np\circ\gamma$ implies by induction $\np\circ\delta\complyR\np\circ\gamma$, so also condition~(\ref{c2}) of Definition~\ref{ccr} is satisfied. 
The proof for rule $(\oplus \cdot+)$ is similar.
\end{proof}
\begin{figure}[t]
\begin{center}
\line(1,0){450}
\end{center}
\[
\begin{array}{c@{\hspace{6mm}}c
}
\prooftree \text{ either } \delta=\gamma=\circ \text{ or }
\Gamma\der \np\circ\delta\complyF\np\circ\gamma
\justifies
\Gamma\der \np\delta\stopA \complyF \np\gamma\sigma
\using(\CkptcomplAx
)
\endprooftree &
\prooftree
\justifies
\Gamma, \np\delta\rho\complyF\np\gamma\sigma \der \np\delta\rho\complyF\np\gamma\sigma 
\using(\CkptcomplHyp)
\endprooftree
\end{array}
\]

\vspace{6mm}

$$
\begin{array}{c}
\prooftree  
    \forall j\in J.~ \Gamma'\der \np{\p\delta\rho}{\rho_j}\complyF\np{\p\gamma\sigma}{\sigma_j}
   \qquad \text{ either } \delta=\gamma=\circ \text{ or }\Gamma'\der  \np\circ\delta\complyF\np\circ\gamma
\justifies
    \Gamma\der \np\delta\rho\complyF\np\gamma\sigma
\using( + \cdot \oplus)
\endprooftree 
\\[8mm]
\multicolumn{1}{l}
{\mbox{ where }\Gamma' = \Gamma,\; \np\delta\rho\complyF\np\gamma\sigma \mbox{ and } \rho = {\whiteckpt_1(\sum_{i\in I\cup J} a_i.{\rho}_i)} 
\mbox{ and } \sigma = {\whiteckpt_2(\bigoplus_{j\in J} \Dual{a}_j.{\sigma}_j)}}\\
\\[10mm]
\prooftree 
    \forall i\in I.~ \Gamma'\der \np{\p\delta\rho}{\rho_i}\complyF\np{\p\gamma\sigma}{\sigma_i}
   \qquad \text{ either } \delta=\gamma=\circ \text{ or }\Gamma'\der  \np\circ\delta\complyF\np\circ\gamma\justifies
    \Gamma\der \np\delta\rho\complyF\np\gamma\sigma
\using(  \oplus \cdot+)
\endprooftree 
\\[8mm]
\multicolumn{1}{l}{\mbox{ where }\Gamma' = \Gamma,\; \np\delta\rho\complyF\np\gamma\sigma \mbox{ and } 
\rho = {\whiteckpt_1(\bigoplus_{i\in I} a_i.{\rho}_i)} 
\mbox{ and } \sigma = {\whiteckpt_2(\sum_{j\in I\cup J} \Dual{a}_j.{\sigma}_j)}}
\end{array}
$$

\caption{The formal system for checkpoint compliance}
\label{fig:complianceSystem}
\begin{center}
\line(1,0){450}
\end{center}
\end{figure}

\begin{theorem}[Completeness] \label{com}\hfill\\
If ~$\np\delta\rho\complyR\np\gamma\sigma$ and $\np{\delta'}{\rho'}\complyR \np{\gamma'}{\sigma'}$ for all $\np{\delta'}{\rho'}\complyF \np{\gamma'}{\sigma'}\in\Gamma$, ~then~ $\Gamma\der \np\delta\rho\complyF\np\gamma\sigma$.
\end{theorem}
\begin{proof}{\em (Sketch)}
By co-induction on the definition of $\complyR$. If $\np\delta\rho\pp \np\gamma\sigma\not\!\!\ored{\tau}$, then  $\rho=\stopA$ and either $ \delta=\gamma=\circ$ or $\delta,\gamma\in\sbc$ by condition~(\ref{c1}) of Definition~\ref{ccr}. In the second case\\ \centerline{$\np\delta\stopA \pp \np\gamma\sigma\ored{\rlbk}
\np\circ\delta\pp\np\circ\gamma$,} which implies $\np\circ\delta\complyR\np\circ\gamma$ by condition~(\ref{c2}) of Definition~\ref{ccr}. So in all cases axiom $(\CkptcomplAx)$ applies.

If $\np\delta\rho\pp \np\gamma\sigma\ored{\tau}$, then either $\rho = {\whiteckpt_1(\sum_{i\in I\cup J} a_i.{\rho}_i)}$ 
 and $\sigma = {\whiteckpt_2(\bigoplus_{j\in J} \Dual{a}_j.{\sigma}_j)}$ or $\rho = {\whiteckpt_1(\bigoplus_{i\in I} \Dual{a}_i.{\rho}_i)}$
 and $\sigma = {\whiteckpt_2(\sum_{j\in I\cup J} a_j.{\sigma}_j)}$. We consider the first case, the proof for the second case being similar. In this case  \mbox{$\np\delta\rho\pp\np\gamma\sigma\ored{\tau} \np{\p\delta\rho}{\rho_j}\pp\np{\p\gamma\sigma}{\sigma_j}$} for all $j\in J$ by Lemma~\ref{simple}. This implies\\ 
 \centerline{$\np{\p\delta\rho}{\rho_j}\complyR\np{\p\gamma\sigma}{\sigma_j}$} by condition~(\ref{c2}) of Definition~\ref{ccr}. If $\np\delta\stopA \pp \np\gamma\sigma\ored{\rlbk}
\np\circ\delta\pp\np\circ\gamma$ we get also $\np\circ\delta\complyR\np\circ\gamma$ by condition~(\ref{c2}) of Definition~\ref{ccr}. So in all cases rule $( + \cdot \oplus)$ applies.
\end{proof}

The main result of our paper is that the formal system provides a complete axiomatic characterisation of the checkpoint compliance, which leads to an decision procedure for checkpoint compliance:

\begin{theorem}[Main Theorem] The formal system $\der$ characterises checkpoint compliance, i.e. 
\[\rho\complyR \sigma \text{ ~~~iff~~~ }
\der\np\circ\rho\complyF \np\circ\sigma.\]
\end{theorem}

\section{Related work and conclusion}
Since the pioneering work by Danos and Krivine~\cite{DK04}, reversible computations in process algebras have been widely studied. The calculus of~\cite{DK04} adds a distributed monitoring system to CCS~\cite{M89} allowing to rewind computations. 
Phillips and Ulidowski~\cite{PU07} propose a method for reversing process operators that are definable by SOS rules in a general format, using keys to bind synchronised actions together.
A reversible variant of the higher-order $\pi$-calculus is defined in~\cite{LMS10}, using name tags for identifying threads and explicit memory processes. In~\cite{LMSS11} Lanese et al.  enrich the calculus of~\cite{LMS10} with a fine-grained rollback primitive. 
The closest paper to ours 
is ~\cite{TY14}, where Tiezzi and Yoshida study the interplay between reverse computations and session-based interactions. Their calculus uses tags and memories as previous proposals in the literature on reversibility.

As pointed out in~\cite{PU07}, reversibility in process calculi is challenging, since  we cannot distinguish between the processes $a\| a$ and $a.a$ by simply recording the past actions. For this reason both histories and unique identifiers for threads have been used to track information. A key requirement, dubbed {\em causal consistency} in ~\cite{DK04}, is that of undoing only actions if no other action depending on them has been executed (and not undone). Session behaviours overcome all these problems: in fact both the client and the server reduce in a sequential way. This justifies the relative simplicity of our  calculus. 

We plan to investigate whether our approach can be extended to multi-party sessions~\cite{HYC08}, the rational being that the parallelism is limited since the interactions must follow the communication protocols prescribed by global types. The subbehaviour relation induced by our notion of compliance is also worth being thoroughly studied.

\medskip
\noindent
{\bf Acknowledgements} The authors gratefully thank the 
referees  for their numerous constructive remarks.

\medskip

\bibliographystyle{eptcs}
\vspace{-1mm}
\bibliography{session}

\end{document}